\documentclass[11pt,a4]{amsart}
\usepackage{amsmath,amssymb,graphicx,verbatim}
\makeatletter
\renewcommand{\section}{\@startsection{section}{1}{0pt}{20pt}{6pt}{\large\bfseries}}
\makeatother

\usepackage{enumerate}
\usepackage{slashbox}

\numberwithin{equation}{section}

\theoremstyle{plain}
  \newtheorem{thm}{Theorem}[section]
  
  \newtheorem{lemma}[thm]{Lemma}
   \newtheorem{cor}[thm]{Corollary}
\theoremstyle{definition}

\newcommand{\R}{\mathbb{R}}
\renewcommand{\Re}{{\mathfrak{Re}}}

\newcommand{\E}{\mathbb{E}}

\newcommand{\fx}[2]{{}_#1\Psi_#2}

\begin{document}

\title[Absolute ruin in the Ornstein-Uhlenbeck type risk model]{Absolute ruin in the Ornstein-Uhlenbeck type risk model}
\author{R.L. Loeffen}
\address{Weierstrass Institute for Applied Analysis and Stochastics,Mohrenstrasse 39
10117 Berlin, Germany} \email{RonnieLambertus.Loeffen@wias-berlin.de}

\author{P. Patie}
   \address{D\'epartement de Math\'ematiques, Universit\'e Libre de Bruxelles\\
Boulevard du Triomphe,  B-1050, Bruxelles, Belgique.}
\email{ppatie@ac.ulb.be}

\begin{abstract}
We start by showing that the finite-time absolute ruin probability in the classical risk model with constant interest force can be expressed in terms of the transition probability of a positive Ornstein-Uhlenbeck type process, say $\hat{X}$. Our methodology  applies to the case when the dynamics of the aggregate claims process is  a subordinator. From this expression,  we easily deduce necessary and sufficient  conditions for the infinite-time absolute ruin to occur. We proceed by showing that,  under some technical conditions, the transition density of  $\hat{X}$ admits a spectral type representation involving merely the limiting distribution of the process. As a by product, we obtain a series expansions for the finite-time absolute ruin probability. On the way, we also derive, for the aforementioned risk process, the Laplace transform of the first-exit time from an interval from above.  Finally, we illustrate our results by detailing some examples.

\bigskip

\noindent{\it Key words:} Risk theory, absolute ruin, Ornstein-Uhlenbeck type processes, first-passage time, spectral representation.

\bigskip

\noindent{\it 2000 Mathematics Subject Classification:} 60G18, 60G51, 60B52
\end{abstract}

\thanks{Research partially funded by  the AXA research fund for the project Modern Actuarial Risk Theory. Both authors would like to thank H. Albrecher for fruitful discussions. The second author is also grateful to F. Avram for valuable comments and the warm hospitality during his stay at the University of Pau.}

\date{}

\maketitle

\bibliographystyle{plain}

\section{Introduction}
Let $Z=(Z_t, t\geq0)$ be a driftless subordinator defined on a filtered probability space $(\Omega,\mathcal{F},(\mathcal F_t)_{t\geq0},\mathbb{P})$, that is $Z$ is a $\R^+-$valued process with stationary and independent increments.  It is well known that the law of $Z$ is characterized by its Laplace exponent which admits the following L\'evy-Khintchine representation
\begin{equation} \label{eq:le_s}
\phi(\beta)=-\log\mathbb{E}(\mathrm{e}^{-\beta Z_1})=\int_0^\infty (1-\mathrm{e}^{-\beta y})\nu( \mathrm{d}y), \quad \beta\geq0,
\end{equation}
where the L\'evy measure satisfies the integrability condition  $\int_0^\infty (1 \wedge y) \nu( \mathrm{d}y)<\infty$ and we refer to the monograph of Kyprianou \cite{Kyprianou-06} for background on subordinators. By $\widehat{\mathbb{P}}$ we denote the law of the dual of $Z$, i.e. the law of $-Z$ under $\mathbb{P}$.
For $x,r,c\in\mathbb{R}$,  we denote   by $\mathbb{P}_x^{(r,c)}$ the law of the process $X$  defined, for any $t\geq0$, by
\begin{equation} \label{eq:def-X}
X_t= \mathrm{e}^{rt} \left( x+\int_0^t\mathrm{e}^{-rs}\mathrm{d}(cs-Z_s) \right),
\end{equation}
where $Z\sim \mathbb{P}$. Similarly,  $\widehat{\mathbb{P}}_x^{(r,c)}$ stands for the law of the process $X$ as defined in \eqref{eq:def-X} with $Z\sim \widehat{\mathbb{P}}$. We simply write $\mathbb{P}_x^{(r)}$ (resp. $\widehat{\mathbb{P}}_x^{(r)}$) for $\mathbb{P}_x^{(r,0)}$ (resp. $\widehat{\mathbb{P}}_x^{(r,0)}$).  Note that equivalently, $X$ is the unique strong solution to the stochastic differential equation
\begin{equation}
\label{SDE}
\mathrm{d}X_t=(rX_t+c)\mathrm{d}t -\mathrm{d}Z_s,\quad X_0=x.
\end{equation}

The process $X$ under $\mathbb{P}_x^{(0,c)}$ with $c>0$ has been used in the literature to model the reserves of an insurance company, the parameter $c$ standing for the premium rate, the jumps of $Z$ standing for the claims and $x$ standing for the initial value of the reserves. In particular, when $\nu(0,\infty)<\infty$, $X$ under $\mathbb{P}_x^{(0,c)}$  is the classical risk process where the claims arrive according to a Poisson process with intensity parameter $\nu(0,\infty)$ and the claim distribution is given by $\nu(\mathrm{d}x)/\nu(0,\infty)$.
When $r>0$, the process $X$ under $\mathbb{P}_x^{(r,c)}$ has been suggested as a  risk process where  the cost of lending/borrowing money are  taken into account.  In this model the insurer earns (credit) interest when the surplus is positive and when the surplus becomes negative, the insurer can cover the deficit by a loan for which he has to pay a (debit) interest. Although in practice the debit interest is much higher than the credit interest, we restrict ourselves to the case where both rates are equal since this choice is particularly tractable and allows us to use techniques which can no longer be used in the general case. Processes of the form \eqref{eq:le_s} with $Z$ a general L\'evy process are known in the literature as processes of Ornstein-Uhlenbeck type (for short OU-type) and therefore in this paper we call the  process $X$ under $\mathbb{P}_x^{(r,c)}$ with $r>0$, the OU-type risk process.


 From \eqref{SDE}, it is clear that when the OU-type risk process reaches the interval $(-\infty,-c/r]$, the premium rate $c$ can no longer compensate the interest payments and so the surplus will decrease to minus infinity. Following Gerber  \cite{Gerber-71}, we say that in this case absolute ruin occurs.
This model and in particular the event of absolute ruin, has been the focus of many research in insurance mathematics, its first appearance can be traced back to Segerdahl \cite{Segerdahl-42}. For more recent investigations and substantial refinements, we mention Gerber \cite{Gerber-71},  Dassios and Embrechts \cite{Embrechts-Dassios-89}, Embrechts and Schmidli \cite{Embrechts-Schmidli-94}, Schmidli \cite{schmidli1994}, Sundt and Teugels \cite{Sundt-Teugels-95}, Albrecher et al.~\cite{Albrecher-01}, Gerber and Yang \cite{Gerber-Yang-07} and Cai \cite{Cai-07}. We refer to the survey papers of Paulsen \cite{paulsen_survey1} and \cite{Paulsen-08} for an overview of ruin models with interest. For general background in ruin theory, we refer to Gerber \cite{ge79} and Albrecher and  Asmussen \cite{Albrecher-Asmussen-10}.

We also point out that the process $X$ under $\widehat{\mathbb{P}}_x^{(-r)}$ with $r>0$, is well known in the literature and has appeared in various settings. For instance, in mathematical finance, this process has been used by Barndorff-Nielsen and Shephard for the modeling of  the stochastic volatility of stock prices, see \cite{Barndorff-Shephard-01}. It   also  belongs to the class of one factor affine term structure models, see e.g.~Filipovic~\cite{filipovic-01}.
Moreover, when the L\'evy measure is finite, it is a specific example of a Poisson shot noise process, see e.g. Perry et al. \cite{Perry-01} and  Iksanov and Jurek \cite{iksanovjurek}.

 The remainder of the paper is organized as follows. The next section is devoted to the statement of our main results which are then proved in Section \ref{section:proofs}. Examples illustrating our approach are presented in Section \ref{section:examples}.

\section{Main results}
We start by providing a representation of the law of the absolute ruin time. To this end, we introduce some notation. First, let, for any $a\leq x$,
\begin{equation*}
\tau_{a}=\inf\{s>0:\: X_s<a\}
\end{equation*}
be the first-passage time below the level $a$ for $X$.
  Henceforth, we shall assume that $r>0$.
We are interested in the  distribution  of $\tau_{-c/r}$ under $\mathbb{P}_x^{(r,c)}$.  We call this random variable the absolute ruin time. As mentioned in the introduction, the reason for the adjective `absolute' is that once the OU-type risk process  goes below the level $-c/r$, it will never go back above this critical level, i.e.
\begin{equation}
\label{def_absoluteruin}
\mathbb{P}^{(r,c)}_x \left(X_t\geq -c/r, t>\tau_{-c/r}\right)=0.
\end{equation}
We first note that from \eqref{eq:def-X} we immediately see that the process $X$ under $\mathbb{P}_x^{(r,c)}$ has the same law as the process $(X_t-c/r,t\geq0)$ under $\mathbb{P}_{x+c/r}^{(r)}$. In particular
\begin{equation*}
\mathbb{P}_x^{(r,c)} \left(\tau_{-c/r}\in\mathrm{d}t\right) = \mathbb{P}_{x+c/r}^{(r)} \left(\tau_0\in \mathrm{d}t\right), \quad t\geq0.
\end{equation*}
Based on this observation, we state all the results in the paper for the $c=0$ case only; the analogue for $c\neq0$ is then obvious. 
The first theorem gives the link between the distribution of $\tau_0$ under $\mathbb{P}^{(r)}_x$ and the distribution of $X$ under $\widehat{\mathbb{P}}^{(-r)}_0$, which leads to an explicit expression for the Laplace transform in space of the finite-time absolute ruin probability. This is in contrast with the the finite-time ruin probability in the $r=0$ case, where only an explicit expression for the double Laplace transform in space and time exists (cf. Theorem 8.1(ii) of Kyprianou \cite{Kyprianou-06}).
\begin{thm}  \label{thm:1}
For any $x>0$ and $t\geq0$, we have
\begin{equation*}
 \mathbb{P}^{(r)}_x\left(\tau_{0} \leq t\right) =  \widehat{\mathbb{P}}^{(-r)}_0\left(X_t> x \right).
\end{equation*}
Consequently, for any $\beta>0$, we have
\begin{equation*}
\mathbb{P}^{(r)}_{{\bf{e}}_{\beta}}(\tau_{0}> t )   =  \exp \left(  -\int_0^t\phi(\beta \mathrm{e}^{-rs}) \mathrm{d}s \right)
 \end{equation*}
where ${\bf{e}}_{\beta}$ stands for the exponential distribution of parameter $\beta>0$ (where we used the notation ~$\mathbb{P}^{(r)}_{{\bf{e}}_{\beta}} (A)=\int_{\R}\mathbb{P}^{(r)}_x (A){\bf{e}}_{\beta}(dx)$).
\end{thm}


Theorem  \ref{thm:1} in combination with Theorem 17.5 of Sato \cite{Sato-99} leads to the following result about the infinite time absolute ruin probability. We refer to Sato \cite{Sato-99} for background on self-decomposable random variables.
\begin{cor} \label{cor:1} \
\begin{itemize}
 \item [(i)] If  $\int_1^\infty \log(y)\nu(\mathrm{d}y)<\infty$, then, under $\widehat{\mathbb{P}}_0^{(-r)}$, $X_t$ converges in distribution, as $t\rightarrow \infty$, to a positive self-decomposable random variable $X_{\infty}$ and
\begin{equation*}
\mathbb{P}^{(r)}_x\left(\tau_{0} <\infty\right) =  \widehat{\mathbb{P}}^{(-r)}_0 \left(X_{\infty}\geq x \right).
\end{equation*}
Moreover, for any $\beta>0$, we have
\[\mathbb{P}^{(r)}_{{\bf{e}}_{\beta}}(\tau_{0}= \infty ) = \exp \left( -\int_0^{\beta}\frac{\phi(u)}{ru}\mathrm{d}u \right).\]

\item[(ii)]
If  $\int_1^\infty\log(y)\nu( \mathrm{d}y) =\infty$, then, for any $x>0$,
\begin{equation*} 
 \mathbb{P}^{(r)}_x\left(\tau_{0} < \infty\right) =  1.
\end{equation*}
\end{itemize}
\end{cor}

It is interesting to note that for the risk process without interest ($r=0$), ruin is certain when the safety loading is negative that is whenever $c\leq \int_1^\infty y\nu(\mathrm{d}y)$, whereas for $r>0$ the premium rate does not have any influence on whether ruin is certain or not.

\bigskip

Before stating  the next two theorems we need to introduce a little further notation. Let  $\xi=\mathbb{I}_{\{\int_1^\infty \log(y)\nu(\mathrm{d}y)=\infty\}}$ and define the measure $W$ on $[0,\infty)$ via its Laplace transform as follows
\begin{equation}\label{eq:lt_w}
\int_0^\infty \mathrm{e}^{-\beta x} W(\mathrm{d}x) = \mathrm{e}^{- \int_{\xi}^\beta \frac{\phi(u)}{ru} \mathrm{d}u}, \quad \beta\geq \xi.
\end{equation}
Noting, by an integration by parts, that
\begin{equation*}
\int_{\xi}^\beta \frac{\phi(u)}{ru} \mathrm{d}u = \int_0^\infty (1-\mathrm{e}^{(\xi-\beta) x}) \frac{\mathrm{e}^{-\xi x}\nu(x,\infty)}{rx} \mathrm{d}x,
\end{equation*}
 it follows readily that
\[W(\mathrm{d}x)=\mathrm{e}^{\xi x} \mu(\mathrm{d}x),\] where $\mu$ is the law of a positive, self-decomposable random variable with infinite L\'evy measure $\frac{\mathrm{e}^{-\xi x}\nu(x,\infty)}{rx} \mathrm{d}x$ (cf. \cite[Corollary 15.11]{Sato-99}). It follows that the function $W(x):=W[0,x]$ is well-defined, increasing and by \cite[Theorem 27.4]{Sato-99} that $W(x)$ is continuous. Note that if $\int_1^\infty \log(y)\nu(\mathrm{d}y)<\infty$, then $W$ is simply the distribution function of the random variable $X_\infty$ from Corollary \ref{cor:1}. We extend $W$ to the whole real line by setting $W(x)=0$ on $(-\infty,0)$.
Further, for all $n<\nu(0,\infty)/r$, the $n$-th derivative of $W$, denoted by $W^{(n)}$, exists and is continuous on $(-\infty,\infty)$ (cf. \cite[Theorem 28.4]{Sato-99}).


The next theorem concerns a discrete spectral type representation of the transition distribution  of the process $X$ under $\widehat{\mathbb{P}}_0^{(-r)}$. We stress  that the spectral theory for self-adjoint operators  in an Hilbert space structure is well established. In particular, McKean~\cite{McKean-56} discusses in details the nature of the spectrum of the semigroup of linear diffusions. However, this general theory does not apply here since we are dealing with non self-adjoint operators.  In this context, there are little examples in the literature where such a spectral decomposition has been given. One notable exception is the paper of  Ogura \cite{Ogura-70}. Therein, the author  provides conditions under which the semigroup of continuous state branching processes with immigration (for short CBI) admits a discrete or continuous spectral representation. Furthermore, in a very elegant fashion, he manages to characterize through Laplace transform, the eigenmeasure and eigenfunctions associated to the  semigroup. To be more precise, in \cite[Theorem 3.1]{Ogura-70}, Ogura shows, under some conditions, that for some $t_0\geq0$ the semigroup  $P_t(x,\mathrm{d}y)$  of a CBI process satisfies the discrete spectral representation
\begin{equation}
\label{spectral_representation}
P_t(x,\mathrm{d}y)= \sum_{n=0}^\infty \eta_n(x) \zeta_n(\mathrm{d}y)\mathrm{e}^{-\lambda_n t}, \quad t\geq t_0, \: x,y\geq0,
\end{equation}
where $\eta_n$, resp.~$\zeta_n$, are eigenfunctions, resp.~eigenmeasures, of $P_t$ corresponding to the eigenvalue $\mathrm{e}^{-\lambda_n t}$ with $\lambda_n\geq0$.
 Although the process $X$ under $\widehat{\mathbb{P}}_x^{(-r)}$ does belong to the class of CBI processes, they were excluded in \cite{Ogura-70}. Moreover, the  methodology of \cite{Ogura-70} does not extend to our case; in particular the suggestion made in footnote 3) on p.309 of \cite{Ogura-70} does not lead to the right direction. In Theorem \ref{thm:sr} below we are able to give, under some technical conditions on the L\'evy measure, a spectral representation similar to  \eqref{spectral_representation} for the semigroup of this process when $x=0$. 
We refer to Chapter 2.7 of Bingham et al.~\cite{Bingham-Goldie-Teugels-89} for the definition of  a quasi-monotone function and remark  that the theorem remains valid if  quasi-monotone is replaced by ultimately monotone. Recall that a function $f:(0,\infty)\rightarrow (0,\infty)$ is ultimately monotone at infinity if there exists $A>0$ such that $f$ is monotone on $(A,\infty)$ and that $f$ is slowly varying at infinity if $\lim_{x\rightarrow\infty} \frac{f(tx)}{f(x)}=1$ for all $t>0$.

 \begin{thm} \label{thm:sr}
 Assume that the L\'evy measure $\nu$ satisfies 
 \begin{enumerate}[i)]
 \item $\int_1^\infty \mathrm{e}^{-\beta x}\nu(\mathrm{d}x)<\infty$ for all $\beta\in\mathbb R$, \item $\nu(x,\infty)= \ell(1/x)x^{-\alpha}$ with $0<\alpha<1$ and $\ell$ a quasi-monotone slowly varying function at infinity. 
     \end{enumerate}
 Then for any $t>t_\alpha=-\frac1{r\alpha}\log\cos\left(\frac{\pi\alpha}2\right)$  and $x\in\mathbb{R}$, the transition distribution  of the process  $X$ under $\widehat{\mathbb{P}}^{(-r)}_x$  is given by
\begin{equation*}
\widehat{\mathbb P}^{(-r)}_x\left(X_t\in\mathrm{d}y\right)=\sum_{n=0}^\infty \mu_n \mathrm{e}^{-rnt} W^{(n+1)}(y-x\mathrm{e}^{-rt})\mathrm{d}y, \quad y\in\mathbb R,
\end{equation*}
 where
\begin{equation*}
\mu_n=\widehat{\mathbb{E}}^{(-r)}[X_{\infty}^n]/n!= \frac1{n!} \left( \frac{\mathrm{d}^n}{\mathrm{d}v^n} \exp \left( -\int_0^{v}\frac{\phi(u)}{ru}\mathrm{d}u \right) \right)_{v=0}.
\end{equation*}
Consequently, for any $x>0$,
 \begin{equation} \label{eq:main}
 \mathbb{P}^{(r)}_{x}(\tau_{0}> t )  = \sum_{n=0}^\infty \mu_nW^{(n)}(x)\mathrm{e}^{-rnt}, %
 \quad  t>t_{\alpha}.
 \end{equation}
Moreover, for any $n$, $W^{(n+1)}(y)\mathrm{d}y$ is an eigenmeasure of $\widehat{\mathbb P}^{(-r)}_x(X_t\in\mathrm{d}y)$ corresponding to the eigenvalue $\mathrm{e}^{-rn t}$ in the sense that
 \begin{equation*}
 \int_{-\infty}^\infty \widehat{\mathbb P}^{(-r)}_x(X_t\in\mathrm{d}y) W^{(n+1)}(x)\mathrm{d}x = \mathrm{e}^{-rn t} W^{(n+1)}(y)\mathrm{d}y, \quad y\in\mathbb R.
 \end{equation*}
 \end{thm}

Our final theorem concerns the two-sided exit problem for the OU type process. More specifically, we compute the Laplace transform of the stopping time
\begin{equation*}
\tau_a^+=\inf\{t>0: X_t>a\},
\end{equation*}
on the event that $X$ exits the interval $[0,a]$ at $a$.
When $X$ is a spectrally negative L\'evy process, this quantity is given by $f_q(x)/f_q(a)$ with $q\geq0$ being the parameter of the Laplace transform and where $f_q$ (resp.~$f_0$) is the so-called $q$-scale  function (resp.~scale function) of the spectrally negative L\'evy process, cf. \cite[Theorem 8.1]{Kyprianou-06}. In Theorem \ref{thm:scale}, we derive a similar expression for the case  where $X$ is the OU-type risk process and show that the `corresponding $q$-scale function' is given in terms of the fractional integral of  the function $W$, which we denote by $W_q$. Hence if one has an explicit expression for $W$, one automatically gets $W_q$ for $q>0$ in closed-form; this is in contrast to the case of spectrally negative L\'evy processes, see the discussion on p.1674 of Kyprianou and Rivero \cite{Kyprianou-Rivero-08}.

\begin{thm}\label{thm:scale}
For $q\geq0$, define the function $W_q:(-\infty,\infty)\rightarrow[0,\infty)$ by $W_q(x)=0$ on $(-\infty,0)$ and on $[0,\infty)$ by the Riemann-Liouville fractional integral of order $q$ of $W,$ i.e.
\begin{equation}\label{eq:dfi}
 W_{q}(x)= \frac{1}{\Gamma(q)}\int_0^x(x-y)^{q-1}W(y)\mathrm{d}y, \quad x\geq0.
\end{equation}
Then, for any $q\geq 0$, $x\leq a$ and $a>0$, 
we have
\begin{equation*}
\mathbb{E}^{(r)}_x \left[  \mathrm{e}^{-q\tau_a^+} \mathbf{1}_{\{\tau_a^+ <\tau_0\}} \right] = \mathbb{E}^{(r)}_x \left[  \mathrm{e}^{-q\tau_a^+} \mathbf{1}_{\{\tau_a^+ <\infty\}} \right] = \frac{ W_{q/r}(x)}{ W_{q/r}(a)}.
\end{equation*}
\end{thm}
We point out that Theorem \ref{thm:scale} simultaneously gives the first-passage time above $a$ of the OU-type risk process. For a process $X$ defined by \eqref{eq:def-X} with $Z$ a general L\'evy process with no negative jumps  and with $r<0$, Hadjiev \cite{Hadjiev-84} (under an extra assumption) and Novikov \cite{Novikov-03} provided an explicit expression for the Laplace transform of the first-passage time of $X$ above a fixed level. Although not considered in these papers, it can be checked that their methodology and expression extend to the case where $r>0$ and $Z$ a L\'evy process with no downward jumps and paths of unbounded variation. However, when $r>0$ and $Z$ is a subordinator (that is, the OU-type risk process), the situation is completely different due to the possibility of absolute ruin.
We remark here that we cannot expect to find a more explicit expression like the one given in \cite{Hadjiev-84} and \cite{Novikov-03} for the $r<0$ case, since this would lead to an explicit expression for the distribution function of any positive self-decomposable random variable.

%

\section{Examples}\label{section:examples}
Before studying some examples let us recall the definition of a transformation recently introduced by  Kyprianou and Patie \cite{Kyprianou-Patie-08}, which will be helpful in our context. For any Laplace exponent $\phi$ of a subordinator, that is of the form \eqref{eq:le_s}, we write for any $\gamma\geq 0$
\[\mathcal{T}_{\gamma}\phi(\beta)=\frac{\beta}{\beta+\gamma}\phi(\beta+\gamma).\]
Then, in \cite{Kyprianou-Patie-08}, it is proved that $\mathcal{T}_{\gamma}\phi$ is  the Laplace exponent of the driftless subordinator with L\'evy measure $\nu_\gamma(\mathrm{d}x)=\mathrm{e}^{-\gamma x}(\nu(\mathrm{d}x)+\gamma\nu(x,\infty)\mathrm{d}x)$. In particular, $\nu_\gamma(x,\infty)=\mathrm{e}^{-\gamma x}\nu(x,\infty)$. Obviously, $\mathcal{T}_0\phi=\phi$. We assume throughout this section that $\int_1^\infty \log(y)\nu(\mathrm{d}y)<\infty$ and therefore the quantity $\varphi_r(\beta)$ defined by
\begin{equation*}
\varphi_{r}(\beta) = \frac{1}{r}\int_0^{\beta}\frac{\phi(s)}{s}\mathrm{d}s,
\end{equation*}
is well-defined for $\beta\geq0$. It is easily seen, after a change of variable, that
\begin{equation*}
\varphi^{(\gamma)}_{r}(\beta) := \frac{1}{r}\int_0^{\beta}\frac{\mathcal{T}_{\gamma}\phi(s)}{s}\mathrm{d}s \\
= \varphi_{r}(\beta+\gamma)-\varphi_{r}(\gamma).
\end{equation*}
In other words, the action of the mapping $\mathcal{T}_\gamma$ on the backward Laplace exponent is equivalent to the action of the Esscher transform on the Laplace exponent of the limiting distribution.
In particular, if $W(x;\gamma)$  stands for the limiting distribution function associated to the backward Laplace exponent $\mathcal{T}_{\gamma}\phi$, then we have the following simple relationship
\begin{equation} \label{eq:Tgamma}
W'(x;\gamma) = \mathrm{e}^{-\varphi_{r}(\gamma)-\gamma x} W'(x;0), \quad x>0.
\end{equation}
(Note that by \cite[Theorem 28.4]{Sato-99}, $W$ is always differentiable on $(0,\infty)$.)

Next we detail some examples where the function $W$ and/or the distribution of the absolute ruin time can be given in closed form. For more examples of cases where $W$ is explicit, we refer to Iksanov and Jurek \cite{iksanovjurek}. We end this section by giving an example that illustrates Theorem \ref{thm:sr}.

\subsection{The compound Poisson case with exponential jumps}
Let us start with the case where for $x>0$,
\begin{equation*}
\nu(x,\infty)=\eta\mathrm{e}^{-\delta x}, \quad \delta,\eta>0.
\end{equation*}
Hence the subordinator $Z$ is a compound Poisson process with exponentially distributed jumps.
We have
\begin{equation*}
\varphi_r(\beta)= \frac{\eta}r\log \left( 1+\beta/\delta  \right)
\end{equation*}
and so the Laplace transform of $X$ under $\widehat{\mathbb{P}}^{(-r)}_x$ is easily computed via \eqref{laplacetransform}. By inverting this Laplace transform, Perry et al. \cite{Perry-01} shows that the transition function of $X$ under $\widehat{\mathbb{P}}^{(-r)}_0$ admits the following form
\begin{equation*}
\widehat{\mathbb{P}}^{(-r)}_0 (X_t \in \mathrm{d}y) = \mathrm{e}^{-\eta t} \delta_0(\mathrm{d}y) + \frac{\eta \delta}{r}\mathrm{e}^{-\eta t}(\mathrm{e}^{rt}-1) \mathrm{e}^{-\delta x} {}_1F_1\left(1-\frac{\eta}{r};2;\delta(1-\mathrm{e}^{rt})y\right)\mathrm{d}y,\quad y>0,
\end{equation*}
where $\delta_0$ stands for the dirac point mass  at $0$ and  ${}_1F_1$ is the confluent hypergeometric function, see e.g.~Lebedev \cite[Section 9.9]{Lebedev-72}. Hence, from Theorem \ref{thm:1}, we deduce that
\begin{equation*}
{\mathbb{P}}^{(r)}_x (\tau_{0}>t) = \mathrm{e}^{-\eta t} \left(1 + \frac{\eta \delta}{r}(\mathrm{e}^{rt}-1)\int_0^{x}\mathrm{e}^{-\delta y}{}_1F_1\left(1-\frac{\eta}{r};2;\delta(1-\mathrm{e}^{rt})y \right) \mathrm{d}y\right).
\end{equation*}
From Corollary \ref{cor:1}, it follows that ${\mathbb{P}}^{(r)}_x (\tau_{0}=\infty)$ (or equivalently $W(x)$) equals the gamma distribution with shape parameter $\eta/r$ and scale parameter $1/\delta$; this fact was first established by Gerber \cite{Gerber-71}.

\subsection{The compound Poisson process with Linnik distribution }
 We consider an example found in Iksanov and Jurek \cite{iksanovjurek}. Assume that the L\'evy measure $\nu$ of $Z$ is given by
\begin{equation*}
\nu(x,\infty)=\eta E_\alpha(-\delta x^\alpha ), \quad x>0,
\end{equation*}
where $ \eta,\delta>0, 0<\alpha\leq1$ and
\[E_{\alpha}(x)=\sum_{k=0}^\infty \frac{x^k}{\Gamma(1+\alpha k)}\] is the Mittag-Leffler function and $\Gamma$ is the gamma function.
Hence the subordinator $Z$ is a compound Poisson process with arrival rate $\eta$ and with jumps distributed according to a  positive Linnik distribution. We have $\phi(\beta)=\frac{\eta \beta^\alpha}{\delta+\beta^\alpha}$ and thus
\[\varphi_r(\beta)=\frac{\eta}{r\alpha}\log(1+\beta^\alpha/\delta).\]
We point out that in the case $\alpha=1$, the Linnik distribution boils down to the exponential distribution and hence this example can be seen as a generalization of the previous. Next, we deduce, from the identity \eqref{laplacetransform} below, that
\begin{equation*} 
\widehat{\E}^{(-r)}_0\left[ \mathrm{e}^{-\beta X_t}\right] =  \left(\frac{\delta^{-1}\beta^{\alpha} \mathrm{e}^{-\alpha rt}+1}{\delta^{-1}\beta^{\alpha}+1}\right)^{\frac{\eta}{\alpha r}}.
\end{equation*}
Denote for $\kappa>0$ by  ${}_{\kappa}W(x)$ the increasing function on $[0,\infty)$ characterized through its Laplace transform which is given by
\begin{equation*}
\int_0^\infty \mathrm{e}^{-\beta x}{}_{\kappa}W(\mathrm{d}x)=\left({\delta^{-1}\beta^{\alpha}+1}\right)^{-\kappa}, \quad \beta\geq0.
\end{equation*}
Note that  ${}_{\frac{\eta}{\alpha r}}W(x)=\mathbb{P}_x^{(r)}(\tau_0=\infty)$.
By means of the binomial
formula, we get, for any $\beta^{\alpha}>1$,
\begin{eqnarray*}
(1+\delta^{-1}\beta^{\alpha})^{-\kappa} &=& \sum_{n=0}^{\infty} (-1)^{n}
\frac{\Gamma(\kappa+n)}{n!\Gamma(\kappa)}
(\delta^{-1}\beta^{\alpha})^{-(n+\kappa)},
\end{eqnarray*}
and hence we obtain by a term-by-term inversion
\begin{eqnarray*}
{}_{\kappa}W'(x)  &=& \delta^{\kappa} x^{\alpha\kappa-1}\sum_{n=0}^\infty (-1)^n\frac{\Gamma(\kappa+n)}{\Gamma(\kappa)\Gamma(\alpha(n+\kappa))n!}(\delta x^{\alpha})^{n}\\
&=& \delta^{\kappa} x^{\alpha\kappa-1} \fx{1}{1} \left( \left.\begin{array}{c}
                  (1,\kappa) \\
                  (\alpha,\alpha\kappa
                  )
                \end{array} \right|
 \: - \delta x^{\alpha}\right)
 \end{eqnarray*}
where $\fx{1}{1}$ is a Wright hypergeometric function, see e.g.~Braaksma \cite[Chap.~12]{Braaksma-64}. Note that for $\alpha=1$,
\begin{eqnarray*}
{}_{\kappa}W'(x)  &=& \frac{\delta^{\kappa}}{\Gamma(\kappa)} x^{\kappa-1}\mathrm{e}^{-\delta x},
 \end{eqnarray*}
recovering the exponential jumps case studied above. Next, observing that for $\beta$ big enough by the binomial theorem,
\begin{equation*}
\begin{split}
\left(\frac{\delta^{-1}\beta^{\alpha} \mathrm{e}^{-\alpha rt}+1}{\delta^{-1}\beta^{\alpha}+1}\right)^{\frac{\eta}{\alpha r}}
= & \left( 1-\frac{1-\mathrm{e}^{-\alpha rt}}{\delta^{-1}\mathrm{e}^{-\alpha rt}\beta^{\alpha}+1} \right)^{-\frac{\eta}{\alpha r}} \\
= & \mathrm{e}^{-\eta t} + \sum_{n=1}^\infty \frac{\Gamma\left(n+\frac{\eta}{\alpha r}\right)}{\Gamma\left(\frac{\eta}{\alpha r}\right)n!} \frac{(1-\mathrm{e}^{-\alpha rt})^n} {(1+\delta^{-1}\mathrm{e}^{-\alpha rt}\beta^{\alpha})^n} ,
\end{split}
\end{equation*}
we get, by means of Laplace transform inversion and Theorem \ref{thm:1},
\begin{equation}\label{eq:ser}
\mathbb{P}^{(r)}_x (\tau_{0}>t) = \mathrm{e}^{-\eta t} \left( 1 +
 \mathrm{e}^{rt}\sum_{n=1}^\infty \frac{\Gamma\left(n+\frac{\eta}{\alpha r}\right)}{\Gamma\left(\frac{\eta}{\alpha r}\right)n!}  (1-\mathrm{e}^{-\alpha rt})^n {}_{n}W(\mathrm{e}^{rt}x) \right).
\end{equation}
Note that, for any $\kappa>0$, the power series
\[\sum_{n=0}^\infty \frac{\Gamma(\kappa+n)}{n!}x^n\]
is analytic in the disc of radius $1$. Since for any $t>0$, $|1-\mathrm{e}^{-rt}|<1$ and, clearly  ${}_nW(y)\leq 1$ for any $y>0$ and $n=1,2\ldots$, we deduce that the series on the right-hand side of  \eqref{eq:ser} is uniformly convergent on $\R^+\times\R^+$.
Note also that the identity \eqref{eq:ser} simplifies when $\frac{\eta}{\alpha r}$ is an integer. In particular, when ${\eta}={\alpha r}$, the finite-time ruin probability is simply a weighted sum of $1$ and the infinite-time ruin probability, i.e.
\begin{equation*}
\mathbb{P}^{(r)}_x (\tau_{0}>t) = \mathrm{e}^{-\eta t} + (1-\mathrm{e}^{-\eta t}){}_{1}W(x).
\end{equation*}
Moreover, when $\alpha=1$, we have
\begin{equation*}
\begin{split}
 \mathrm{e}^{rt}\sum_{n=1}^\infty \frac{\Gamma\left(n+\frac{\eta}{  r}\right)}{\Gamma\left(\frac{\eta}{  r}\right)n!}  (1-\mathrm{e}^{-rt})^n {}_{n}W(\mathrm{e}^{rt}x) = &
\int_0^x \sum_{n=1}^\infty \frac{\Gamma\left(n+\frac{\eta}{  r}\right)}{\Gamma\left(\frac{\eta}{  r}\right)n!} (\mathrm{e}^{rt}-1)^n  \frac{\delta^{n}}{\Gamma(n)} y^{n-1}\mathrm{e}^{-\delta\mathrm{e}^{rt} y}\mathrm{d}y \\
= & \delta(\mathrm{e}^{rt}-1)\int_0^x \mathrm{e}^{-\delta \mathrm{e}^{rt} y} {}_1F_1 \left( 1+\frac{\eta}{r};2;(\mathrm{e}^{rt}-1)\delta y \right)  \mathrm{d}y  \\
= & \delta(\mathrm{e}^{rt}-1)\int_0^x \mathrm{e}^{-\delta y} {}_1F_1 \left( 1-\frac{\eta}{r};2;(1-\mathrm{e}^{rt})\delta y \right) \mathrm{d}y
\end{split}
\end{equation*}
where the last equality follows from the identity $\mathrm{e}^{-y}{}_1F_1\left(a;b;y\right)={}_1F_1\left(b-a;b;-y\right)$, see,  after correcting the obvious typo, formula 9.11.2 in \cite{Lebedev-72}. This is  consistent with the previous example. We finally note that when $0<\alpha<1$, we can use \eqref{eq:Tgamma} to get a new example where the infinite-time ruin probability is explicit.

\subsection{The stable subordinator}
We assume that $Z$ under $\mathbb{P}$ is a stable subordinator of index $0<\alpha<1$.  Its
Laplace exponent takes the form
$\phi(\beta)=\frac{\beta^{\alpha}}{\cos(\pi\alpha/2)}$, $\beta \geq0$.
It is easy to verify that in this case
$\varphi_r(\beta) =\frac{ \beta^{\alpha}}{\cos(\pi\alpha/2)\alpha r}$, $\beta \geq0$.
One can check by taking Laplace transforms and using the scaling property of $Z$ that,
\begin{equation*}
\widehat{\mathbb{P}}^{(-r)}(X_t\in\mathrm{d}y) = {\mathbb{P}} \left( v^{1/\alpha}(t)Z_{1} \in\mathrm{d}y \right) \quad \text{where $v(t) =\frac{1-\mathrm{e}^{-\alpha rt}}{\alpha r}$.}
\end{equation*}
Hence, by writing $P_{\alpha}$ for the distribution of a positive stable random variable of index $\alpha$, we deduce  in combination with Theorem \ref{thm:1} that
\begin{equation*}
{\mathbb{P}}^{(r)}_x (\tau_{0}>t) = P_{\alpha}\left( xv^{-1/\alpha}(t) \right)
\end{equation*}
and by taking the limit as $t\rightarrow \infty$, we get from Corollary \ref{cor:1},
\begin{equation*}
{\mathbb{P}}^{(r)}_x (\tau_{0}=\infty) = P_{\alpha}\left((\alpha r)^{1/\alpha}x \right).
\end{equation*}
We mention that a series representation of stable densities can be found in Sato \cite{Sato-00}. Additional interesting representations has also been derived by Schneider \cite{Schneider-86}. In particular,  in the case $\alpha=1/2$, we recall that the stable distribution admits the well known simple expression
\begin{equation}
P_{1/2}(x)=\int_0^{x}\frac{1}{\sqrt{2\pi}} y^{-3/2}\mathrm{e}^{-1/(2y)}\mathrm{d}y.
\end{equation}
We consider now the image of $\phi$ by the mapping $\mathcal{T}_{\gamma}$ which gives
\begin{eqnarray*}
\mathcal{T}_{\gamma} \phi(\beta) = \frac{(\beta+\gamma)^{\alpha}-\gamma(\beta+\gamma)^{\alpha-1}}{\cos(\pi\alpha/2) }.
\end{eqnarray*}
This Laplace exponent is a specific instance
of the family of characteristic exponents of truncated L\'evy processes constructed by Boyarchenko and Levendorskii \cite{Boyarchenko-L-00}, also called tempered stable processes.
From the identity \eqref{eq:Tgamma}, we deduce that
\[W(x;\gamma)=\mathrm{e}^{-\gamma^{\alpha}/(\cos(\pi\alpha/2) \alpha r)}(\alpha r)^{1/\alpha}\int_0^x \mathrm{e}^{-\gamma y} P'_{\alpha}\left( (\alpha r)^{1/\alpha} y\right) \mathrm{d}y.\]

\subsection{Example from Theorem \ref{thm:sr}}

We now give an example where we use the representation \eqref{eq:main} for the finite-time absolute survival probability in Theorem \ref{thm:sr}. Assume $\nu(x,\infty)=\ell(1/x)x^{-\alpha}$ with $0<\alpha<1$ and
\begin{equation*}
\ell(x)= C \mathbb I_{\{x\geq A\}}, \quad \text{where $C,A>0$.}
\end{equation*}
Then $\nu$ is a L\'evy measure that satisfies all the conditions of Theorem \ref{thm:sr}. Note that this L\'evy measure corresponds to a subordinator with stable-like jumps near the origin and no jumps larger than $A$. This model with claims bounded in size, might for instance be used in the case where the insurance company  has excess-of-loss reinsurance. A straightforward computation involving integration by parts, shows that
\begin{equation*}
\varphi_r(\beta)= \frac{C}{r\alpha}  \left( \beta^\alpha \Gamma(1-\alpha;0,\beta A) - A^{-\alpha} \left(1-\mathrm{e}^{-\beta A}\right) \right),
\end{equation*}
where  $\Gamma(a;z_0,z_1) = \int_{z_0}^{z_1}\mathrm{e}^{-t}t^{a-1}\mathrm{d}t$ is the incomplete gamma function. We want to check how well the spectral representation \eqref{eq:main} performs.
Since we do not know an explicit expression for $W$, we proceed by numerical inversion. Let $N\in\mathbb Z_{+}$. For $n=0,1,2\ldots,N-1$ we compute $W^{(n+1)}$ by inverting the right-hand side of \eqref{nderiv} below using the (inverse) fast Fourier transform. We then numerically integrate $W'$ to get $W$ and consequently we get an approximation for the right-hand side of \eqref{eq:main} where the upper boundary $\infty$ is replaced by $N$.
We remark here that as $n$ increases, it gets more difficult to obtain accurate approximations for $W^{(n+1)}$ since these functions oscillate more and more as $n$ increases. In particular, for $n\geq0$, $W^{(n+1)}$ has $n$ distinct zeros, cf. \cite[Theorem 5.1]{satoyamazato}.
In order to get an approximation for the left-hand side of \eqref{eq:main}, note that from \eqref{laplacetransform}, it follows that the Fourier transform of the transition density of $X_t$ under $\widehat{\mathbb P}_0^{(-r)}$ is given by
\begin{equation}
\label{direct_fourier}
\widehat{\mathbb{E}}_0^{(-r)} \left[ \mathrm{e}^{-\mathrm{i}u X_t} \right]
 = \exp \left(  \varphi_r(\mathrm{i}u\mathrm{e}^{-rt}) -\varphi_r(\mathrm{i}u)  \right), \quad u\in\mathbb R,\:t\geq0,
\end{equation}
and by numerical Fourier inversion and integration, one gets (cf. Theorem \ref{thm:1}) an approximation of the finite-time absolute survival probability at a fixed time $t$.
Hence we  obtain numerical approximations of the truncation error
corresponding to \eqref{eq:main} which we define by
\begin{equation*}
e_{N,t}=\max_{i=0,1,\ldots,M} \left|\mathbb{P}^{(r)}_{hi}(\tau_{0}> t )  - \sum_{n=0}^N \mu_nW^{(n)}(hi)\mathrm{e}^{-rnt}\right|.
\end{equation*}
For our example we take $r=0.2$, $\alpha=0.5$, $C=A=1$, $h=0.2$ and $M=125$. All the calculations are done in Matlab.
In Table \ref{table}, numerical approximations of the truncation error are displayed for different values of $N$ and $t$. One sees that convergence takes place for $t=7,10,15$ as $N$ grows and that for a fixed $N$, the error becomes smaller as $t$ grows. We point out that the numbers in the table, besides the truncation error, also consists of integration error due to the Fourier inversion. In particular, the integration error will dominate the truncation error for high values of $N$ (and low values of $t$) due to the highly oscillating nature of $W^{(n)}$ for large values of $n$; this might explain the large value in the table for $N=16$ and $t=5$.
The spectral representation performs badly for $t=3$, but note that $3<t_\alpha\approx 3.46$ and hence this case is not covered by Theorem \ref{thm:sr}. The convergence of $\sum_{n=0}^N \mu_n W^{(n)}(x)\mathrm{e}^{-rnt}$ to $\mathbb{P}^{(r)}_{x}(\tau_{0}> t )$ for $t=7$ is graphically displayed in Figure 1.

We remark that the benefit of computing finite-time absolute survival/ruin probabilities via \eqref{eq:main} is that one can
quickly calculate these probabilities for a whole range of time points $t$, whereas if ones uses the method via \eqref{direct_fourier}, one has  to perform a separate Fourier inversion for each  $t$. Moreover, by using the spectral representation, one can take advantage of \eqref{eq:Tgamma} to quickly obtain the distribution of the absolute ruin time in the case where the tail of the L\'evy measure equals $\mathrm{e}^{-\gamma x}\nu(x,\infty)$ with $\gamma>0$.

\begin{table}[ht]
\centering
\begin{small}
\begin{tabular}{|c|c c c c c|}
\hline
\backslashbox{$N$}{$t$}   &          3 &          5 &          7 &         10 &         15 \\
\hline \hline
         0 &      0.905 &      0.718 &      0.526 &      0.312 &      0.130 \\

         1 &      0.768 &      0.461 &      0.244 &      0.091 &      0.027 \\

         2 &      1.283 &      0.453 &      0.139 &      0.025 &      0.020 \\

         3 &      1.587 &      0.385 &      0.090 &      0.024 &      0.021 \\

         4 &      2.424 &      0.426 &      0.088 &      0.025 &      0.021 \\

         6 &      4.080 &      0.349 &      0.039 &      0.022 &      0.021 \\

         9 &      9.320 &      0.237 &      0.033 &      0.022 &      0.021 \\

        12 &     22.508 &      0.167 &      0.031 &      0.022 &      0.021 \\

        16 &    582.088 &      0.887 &      0.030 &      0.022 &      0.021 \\ \hline

\end{tabular}
\vspace{0.4cm}

\caption{Numerical approximations of $e_{N,t}$ for different values of $N$ and $t$.}
\label{table}
\end{small}
\end{table}

\begin{figure}
\label{figure}
\centering
\includegraphics[scale=.6]{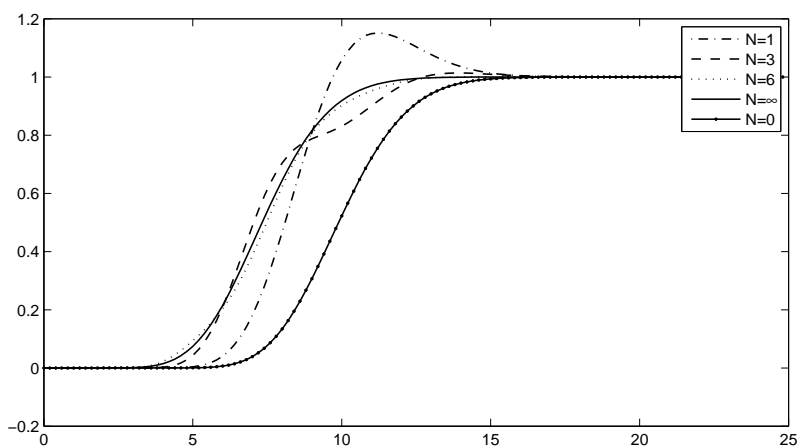}
\begin{small}
\caption{Graph of $x\mapsto\sum_{n=0}^N \mu_n W^{(n)}(x)\mathrm{e}^{-rnt}$ for $t=7$ and  $N=0,1,3,6,\infty$.}
\end{small}
\end{figure}

\section{Proofs}\label{section:proofs}

\begin{proof}[\textbf{Proof of Theorem \ref{thm:1}}]
First, by means of a change of variables and the duality lemma for L\'evy processes, see e.g.~\cite[p.~45]{Bertoin-96}, we obtain
\begin{eqnarray}
\label{little_duality}
\mathbb{P}\left(\int_0^t  \mathrm{e}^{-rs} \mathrm{d}Z_s\in  \mathrm{d}y \right) &=&  \mathbb{P}\left(-\int_0^t  \mathrm{e}^{-r(t-u)} \mathrm{d}Z_{t-u}\in  \mathrm{d}y \right)  \nonumber \\
&=& \mathbb{P}\left(-\int_0^t  \mathrm{e}^{-r(t-u)} \mathrm{d}(Z_{t-u}-Z_t) \in  \mathrm{d}y\right)  \nonumber  \\
&=& \mathbb{P}\left(\int_0^t  \mathrm{e}^{-r(t-u)} \mathrm{d}Z_u \in  \mathrm{d}y \right).
\end{eqnarray}
Let $x>0$. Now using that the ruin time is absolute, \eqref{eq:def-X} and \eqref{little_duality},
\begin{eqnarray*}
 \mathbb{P}_x^{(r)}(\tau_0\leq t) &= & \mathbb{P}_x^{(r)}\left(X_t<0\right)\\
&=&  \mathbb{P}\left(-\int_0^t  \mathrm{e}^{-rs} \mathrm{d}Z_s <-x\right) \\
&= &   \mathbb{P}\left(-\int_0^t  \mathrm{e}^{-r(t-u)} \mathrm{d}Z_{u} <-x\right)\\
&=&   \widehat{\mathbb{P}}^{(-r)}_0\left(X_t>x\right),
\end{eqnarray*}
which gives the first assertion. The second follows  directly from the well-known expression for the Laplace transform of $X_t$ under $\widehat{\mathbb P}^{(-r)}_x$ which reads,
\begin{equation}
\label{laplacetransform}
\widehat{\mathbb{E}}_x^{(-r)} \left[ \mathrm{e}^{-\beta X_t} \right]
 = \exp \left( -\beta x \mathrm{e}^{-rt} +\int_0^t \phi(\lambda\mathrm{e}^{-rs})\mathrm{d}s  \right), \quad \beta,t\geq0, x\in\mathbb R,
\end{equation}
see e.g.~Sato \cite{Sato-99}.
\end{proof}

Denote by $P^{(r)}_t$ (resp.~$\widehat{P}^{(-r)}_t$) the semigroup of the process $X$ under $\mathbb{P}^{(r)}$  (resp.~$\widehat{\mathbb{P}}^{(-r)}$), i.e.~for any  positive measurable function $f$ on $\R$,
\begin{equation*}
P^{(r)}_t f(x)= \mathbb{E}_x^{(r)}\left[f(X_t)\right],\quad x\in \R, t\geq0.
\end{equation*}
For the proof of Theorem \ref{thm:scale} we need the following lemma; note that part (ii) is a (weak) duality result between the two semigroups.
\begin{lemma}\label{lemma_interval} \
\begin{itemize}
 \item[(i)] For any $t\geq0$ and $x\in\R$,
\begin{equation}
\label{asmussen}
(X_t,\mathbb{P}_x^r)\stackrel{(d)}{=}( \mathrm{e}^{rt}x+X_t,\mathbb{P}^r_0),
\end{equation}
where $\stackrel{(d)}{=}$ stands for the identity in distribution.
\item[(ii)] For any two positive measurable functions $f$ and $g$,
\begin{equation} \label{big_duality}
\langle P^{(r)}_tf,g\rangle\: = \: \mathrm{e}^{-rt}\langle f,\widehat{P}^{(-r)}_t g\rangle,
\end{equation}
where  $\langle f,g\rangle =\int_{\mathbb{R}} f(x)g(x) \mathrm{d}x$.

\item[(iii)] For all $a\geq0$, $\mathbb{P}_x^{(r)}\left(X_t\in[-a,a] \textrm{ for all } t\geq0\right)=0$.
\end{itemize}
\end{lemma}
\begin{proof}
Item (i)  follows directly from \eqref{eq:def-X}. Next, using successively Tonelli, \eqref{asmussen}, a change of variables, \eqref{little_duality} and \eqref{asmussen} again, we get
\begin{equation*}
\begin{split}
\langle P^{(r)}_tf,g\rangle  = & \int_{\mathbb{R}}\mathbb{E}^{(r)}_x[f(X_t)]g(x) \mathrm{d}x
= \mathbb{E}^{(r)}_0 \left[ \int_{\mathbb{R}} f(X_t+ \mathrm{e}^{rt}x)g(x) \mathrm{d}x \right] \\
= & \mathbb{E}^{(r)}_0 \left[ \int_{\mathbb{R}}  \mathrm{e}^{-rt} f(z)g((z-X_t) \mathrm{e}^{-rt}) \mathrm{d}z \right] \\
= &   \mathrm{e}^{-rt} \int_{\mathbb{R}}  f(z)\mathbb{E} \left[ g\left(z \mathrm{e}^{-rt} - \int_0^t  \mathrm{e}^{-rs} \mathrm{d}Z_s  \right) \right]  \mathrm{d}z  \\
= &  \mathrm{e}^{-rt} \int_{\mathbb{R}}  f(z)\mathbb{E} \left[ g\left( \mathrm{e}^{-rt} \left(z + \int_0^t  \mathrm{e}^{rs} \mathrm{d}(-Z_s)\right)  \right) \right]  \mathrm{d}z \\
= &  \mathrm{e}^{-rt} \int_{\mathbb{R}}  f(z)\widehat{\mathbb{E}}^{(-r) }_z \left[ g(X_t) \right]  \mathrm{d}z \\
=  & \mathrm{e}^{-rt}\langle f,\widehat{P}^{(-r)}_t g\rangle,
\end{split}
\end{equation*}
yielding (ii). Finally, choosing $f(x)=\mathbf{1}_{\{x\in[-a,a]\}}$ and $g(x)=\mathbf{1}_{\{x\geq0\}}$ with $a\geq0$ in \eqref{big_duality} leads to
$\int_0^\infty \mathbb{P}_x^{(r)}(X_t\in[-a,a])\mathrm{d}x=2a\mathrm{e}^{-rt}$
and now by applying standard argument from measure theory (cf. Exercise 3.1.12 of \cite{stroock}) and Fatou's lemma, we get
\begin{eqnarray*}
\int_0^\infty \mathbb{P}_x^{(r)}(\liminf_{t\rightarrow\infty} \{ X_t\in[-a,a]\} )\mathrm{d}x &\leq& \int_0^\infty \liminf_{t\rightarrow\infty} \mathbb{P}_x^{(r)}(X_t\in[-a,a])\mathrm{d}x \\
&\leq& \liminf_{t\rightarrow\infty} \int_0^\infty \mathbb{P}_x^{(r)}(X_t\in[-a,a])\mathrm{d}x =0.
\end{eqnarray*}
Hence  $\mathbb{P}_x^{(r)}(X_t\in[-a,a]$ for all $t\geq0)=0$ and the proof is completed.
\end{proof}

\begin{proof}[\textbf{Proof of Theorem \ref{thm:scale}}]
The first equality follows directly from Lemma \ref{lemma_interval}(iii) and \eqref{def_absoluteruin}. In order to prove the second equality,
we first show that for all $q\geq0$,
\begin{equation}
\label{invariant}
\mathrm{e}^{-qt} \mathbb{E}_x^{(r)} \left[ W_{q/r} (X_t) \right] =W_{q/r} (x).
\end{equation}
Let $\beta>0$, $q\geq0$ and recall that $\xi=\mathbb{I}_{\{\int_1^\infty \log(y)\nu(\mathrm{d}y)=\infty\}}$. From  the identity \eqref{eq:lt_w} and observing the convolution in \eqref{eq:dfi}, we get
\begin{equation}
\label{eq:lt_wq}
\int_{-\infty}^\infty \mathrm{e}^{-\beta x}W_q(x)\mathrm{d}x  = \beta^{-q-1}\mathrm{e}^{-\int_{\xi}^\beta \frac{\phi(u)}{ru} \mathrm{d}u}.
\end{equation}
On the other hand, using successively \eqref{big_duality}, \eqref{laplacetransform} and \eqref{eq:lt_wq}, we obtain, denoting $e_{\beta}(x)=\mathrm{e}^{-\beta x}$, $x\in\mathbb R$, that the Laplace transform of the left-hand side of \eqref{invariant} is given by
\begin{eqnarray*}
\mathrm{e}^{-qt} \langle e_{\beta},{P}^{(r)}_t W_{q/r}\rangle  &=& \mathrm{e}^{-(q+r)t} \langle \widehat{P}^{(-r)}_te_{\beta},W_{q/r}\rangle  \\
&=& \mathrm{e}^{-(q+r)t} \mathrm{e}^{-\int_{\xi}^\beta \frac{\phi(u)}{ru} \mathrm{d}u + \int_{\xi}^{\beta\mathrm{e}^{-rt}} \frac{\phi(u)}{ru}\mathrm{d}u } \langle e_{\beta \mathrm{e}^{-rt} },W_{q/r}\rangle  \\
&=& \beta^{-q/r-1}  \mathrm{e}^{-\int_{\xi}^\beta \frac{\phi(u)}{ru} \mathrm{d}u}.
\end{eqnarray*}
Hence by uniqueness of the Laplace transform and the continuity of $W_q$, \eqref{invariant} follows.

Now with the aid of \eqref{invariant} we can use the Dynkin formula \cite[Theorem 12.4]{Dynkin-65}, to derive
\begin{equation}
\mathbb{E}^{(r)}_x\left[ \mathrm{e}^{-q(t\wedge \tau_a^+)}W_{q}(X_{t\wedge \tau_a^+}) \right]= W_{q/r} (x).
\end{equation}
Since the mapping $W_{q}$ is increasing and continuous, the proof is completed by an argument of dominated convergence and  the property that under $\mathbb{P}_x^{(r)}$, $X_{\tau_a^+}=a$ on $\{\tau_a^+<\infty\}$ which follows by absence of positive jumps. 
\end{proof}

\begin{proof}[\textbf{Proof of Theorem \ref{thm:sr}}]
Denote $\varphi_r(v)=\int_0^{v}\frac{\phi(u)}{ru}\mathrm{d}u$, $v\in\mathbb C$. By assumption (i), $\varphi_r$ and consequently $\exp(\varphi_r(\cdot))$ are entire functions (see e.g.~\cite[Section 25]{Sato-99})  and hence we can write
\begin{equation}
\label{entire}
\mathrm{e}^{\varphi_r(v)} = \sum_{n=0}^\infty \mu_n v^n, \quad v\in\mathbb{C}.
\end{equation}
By assumption (ii), $\nu(0,\infty)=\infty$ and thus from  Sato and Yamazato \cite[Equation (2.20)]{satoyamazato}, we get
\begin{equation}
\label{nderiv}
W^{(n+1)}(y)=\frac{1}{2\pi}\int_{-\infty}^\infty (\mathrm{i}u)^n \mathrm{e}^{-\varphi_r(\mathrm{i}u)+\mathrm iuy}\mathrm{d}u, \quad n=0,1,2,\ldots.
\end{equation}
From the above identity, we deduce that
\begin{equation}
\label{fourierinversion}
\sum_{n=0}^\infty \mu_n W^{(n+1)}(y)v^n = \frac{1}{2\pi}\sum_{n=0}^\infty \int_{-\infty}^\infty \mu_n (\mathrm{i}u v)^n \mathrm{e}^{-\varphi_r(\mathrm{i}u)+\mathrm iuy}\mathrm{d}u
\end{equation}
and now we want to use Fubini to  switch the sum and the integral.
We have for $v>0$ by \eqref{entire},
\begin{equation}
\label{fubini}
\begin{split}
 \int_{-\infty}^\infty  \sum_{n=0}^\infty  \left|   \mu_n (\mathrm{i}u v)^n \mathrm{e}^{-\varphi_r(\mathrm{i}u)+\mathrm iuy}  \right| \mathrm{d}u
\leq & \sqrt{2}\int_{-\infty}^\infty \sum_{n=0}^\infty  \mu_n(|u| v)^n \mathrm{e}^{\Re(-\varphi_r(\mathrm{i}u))} \mathrm{d}u \\
= & \sqrt{2}\int_{-\infty}^\infty     \mathrm{e}^{\varphi_r(|u|v)} \mathrm{e}^{-\Re(\varphi_r(\mathrm{i}u))} \mathrm{d}u.
\end{split}
\end{equation}

Using Fubini, we have
\begin{eqnarray*}
r\Re(\varphi_r(\mathrm{i}u)) &=& \int_0^\infty (1-\cos(ux))\frac{\nu(x,\infty)}x\mathrm{d}x \\ &=& \int_0^\infty \int_0^u x\sin(xs)\mathrm{d}s \frac{\nu(x,\infty)}x\mathrm{d}x \\
&=&  \int_0^u \int_0^\infty \sin(xs) {\nu(x,\infty)} \mathrm{d}x \mathrm{d}s.
\end{eqnarray*}
Because of condition (ii), it follows by \cite[Equation (4.3.8)]{Bingham-Goldie-Teugels-89} that
\begin{equation*}
\int_0^\infty \sin(xs) {\nu(x,\infty)} \mathrm{d}x \sim s^{\alpha-1} \ell(s)\Gamma(1-\alpha)\cos\left(\frac{\pi\alpha}2\right), \quad \text{as $s\rightarrow\infty$}
\end{equation*}
and hence by Karamata's Theorem \cite[Theorem 1.5.11]{Bingham-Goldie-Teugels-89},
\begin{equation*}
r\Re(\varphi_r(\mathrm{i}u)) \sim u^{\alpha} \ell(u)\frac{\Gamma(1-\alpha)}{\alpha}\cos\left(\frac{\pi\alpha}2\right), \quad \text{as $u\rightarrow\infty$}.
\end{equation*}
Here $f(x)\sim g(x)$ as $x\rightarrow\infty$ stands for $\lim_{x\rightarrow\infty} \frac{f(x)}{g(x)}=1$. Similarly for $u>0$ by Fubini,
\begin{eqnarray*}
r\varphi_r(u)&=&\int_0^\infty (1-\mathrm{e}^{-ux})\frac{\nu(x,\infty)}x\mathrm{d}x \\&=& \int_0^\infty \int_0^u x\mathrm{e}^{-sx}\mathrm{d}s \frac{\nu(x,\infty)}x \mathrm{d}x \\
&= & \int_0^u \int_0^\infty  \mathrm{e}^{-sx} {\nu(x,\infty)} \mathrm{d}x \mathrm{d}s
\end{eqnarray*}
and thus  by Karamata's theorem and Karamata's Tauberian theorem \cite[Theorem 1.7.1]{Bingham-Goldie-Teugels-89},
\begin{equation*}
r\varphi_r(u) \sim \ell(u)\frac{\Gamma(1-\alpha)}{\alpha}u^{\alpha}, \quad \text{as $u\rightarrow\infty$.}
\end{equation*}
In particular $r\varphi_r(u)$ is regularly varying at infinity with parameter $\alpha$ and thus
\begin{equation*}
\lim_{u\rightarrow\infty} \frac{r\varphi_r(vu)}{\ell(u)\frac{\Gamma(1-\alpha)}{\alpha}u^{\alpha}}=
\lim_{u\rightarrow\infty} \frac{r\varphi_r(vu)}{r\varphi_r(u)}\frac{r\varphi_r(u)}{\ell(u)\frac{\Gamma(1-\alpha)}{\alpha}u^{\alpha}}=v^\alpha.
\end{equation*}
It follows that $r[\Re(\varphi_r(\mathrm{i}u))-\varphi_r(vu)]\sim [\cos\left(\frac{\pi\alpha}2\right)-v^\alpha]\ell(u)\frac{\Gamma(1-\alpha)}{\alpha}u^{\alpha}$ as $u\rightarrow\infty$ if $v^\alpha\neq \cos\left(\frac{\pi\alpha}2\right)$.
Assume now that $\cos\left(\frac{\pi\alpha}2\right)>v^\alpha$ and let $\delta=\frac{\alpha}2$. Then $\Re(\varphi_r(\mathrm{i}u)) -\varphi_r(vu)$ is regularly varying at infinity with parameter $\alpha$ and it follows by e.g. \cite[Theorem 1.5.6(iii)]{Bingham-Goldie-Teugels-89} that there exists some $U\geq1$ and $A>0$ such that for all $u\geq U$, $\Re(\varphi_r(\mathrm{i}u))-\varphi_r(vu)\geq Au^{\alpha-\delta}$. Since $\Re(\varphi_r(\mathrm{i}u))$ and $\varphi_r(vu)$ are even and continuous functions, we can now conclude that the right-hand side of \eqref{fubini} is finite.

This allows us to use Fubini in \eqref{fourierinversion} for $v^\alpha<\cos\left(\frac{\pi\alpha}2\right)$ to get with the aid of \eqref{entire},
\begin{equation*}
\begin{split}
\sum_{n=0}^\infty \mu_n W^{(n+1)}(y)v^n = & \frac{1}{2\pi} \int_{-\infty}^\infty \sum_{n=0}^\infty \mu_n (\mathrm{i}u v)^n \mathrm{e}^{-\varphi_r(\mathrm{i}u)+\mathrm iuy}\mathrm{d}u \\
= & \frac{1}{2\pi} \int_{-\infty}^\infty   \mathrm{e}^{\varphi_r(\mathrm{i}vu)-\varphi_r(\mathrm{i}u)+\mathrm iuy}\mathrm{d}u  .
\end{split}
\end{equation*}
The right-hand side of the previous identity is the inverse Fourier transform of the Fourier transform of $\widehat P^{(-r)}_t(0,\mathrm{d}y)$ with $v=\mathrm{e}^{-rt}$. From earlier considerations, we see that this Fourier  transform is integrable when $t>0$. This implies in particular that for $t>0$, $\widehat P^{(-r)}_t(0,\mathrm{d}y)$ has a continuous density $\widehat p^{(-r)}_t(0,y)$, see e.g. the first line of Section 28 in Sato \cite{Sato-99}. Since $\widehat p^{(-r)}_t(0,y)$ is trivially integrable as well, we can now use the Fourier inversion theorem (see e.g.~\cite[Theorem 8.26]{folland}) to conclude
\begin{equation*}
\widehat p^{(-r)}_t(0,y)=\sum_{n=0}^\infty \mu_n W^{(n+1)}(y)\mathrm{e}^{-rnt}, \quad t>-\frac1{r\alpha}\log\cos\left(\frac{\pi\alpha}2\right).
\end{equation*}
The first identity follows then from the identity \eqref{asmussen}.

To prove the second identity, noting that  the right-hand side of \eqref{fubini} does not depend on $y$, we see that for any $f:\mathbb{R}\rightarrow\mathbb{R}$ measurable and integrable (i.e. $\int_{\mathbb{R}}|f(y)|\mathrm{d}y<\infty$), we can use Fubini to deduce
\begin{equation*}
\int_{\mathbb{R}} f(y)\widehat p^{(-r)}_t(0,y)\mathrm{d}y = \sum_{n=0}^\infty \mu_n\mathrm{e}^{-rnt} \int_{\mathbb{R}} f(y) W^{(n+1)}(y)\mathrm{d}y, \quad t>-\frac1{r\alpha}\log\cos\left(\frac{\pi\alpha}2\right).
\end{equation*}
Choosing $f(y)=\mathbf{1}_{\{y\leq x\}}$ and applying Theorem \ref{thm:1} shows \eqref{eq:main}. The last part follows by taking Laplace transforms on both sides, hereby making use of Fubini, \eqref{laplacetransform} and \eqref{eq:lt_w}.  Note that by p.190/191 of \cite{Sato-99},  $W^{(n+1)}(x)$ goes to zero as $x$ goes to $\pm\infty$ and hence $W^{(n+1)}$ is bounded on $(-\infty,\infty)$ which allows one to use Fubini.
\end{proof}

\bibliographystyle{plain}


\end{document}